\documentclass[conference]{IEEEtran}
\usepackage{cite, float, wrapfig, amsmath, amsfonts, eqparbox, array, stfloats, url, xcolor, tabularx, colortbl, booktabs, indentfirst, booktabs,amssymb,amsthm,amscd,dsfont,mathrsfs,comment,url,paralist,caption}
\usepackage[pdftex]{graphicx} 
\DeclareGraphicsExtensions{.pdf,.jpeg,.png}
\usepackage[english]{babel}
\usepackage[tight,footnotesize]{subfigure}
\usepackage[euler]{textgreek}
\restylefloat{table}
\theoremstyle{plain}
\newtheorem{defin}{Definition}
\theoremstyle{plain}
\theoremstyle{plain}
\theoremstyle{plain}
\theoremstyle{plain}
\newtheorem{fact}{Fact}
\theoremstyle{plain}
\theoremstyle{plain}
\newtheorem{lemma}{Lemma}
\theoremstyle{plain}
\theoremstyle{plain}
\theoremstyle{remark}
\newtheorem{remark}{Remark}
\theoremstyle{discussion}
\theoremstyle{plain}

\DeclareMathOperator*{\argmax}{arg\,max} 
\definecolor{light-gray}{gray}{0.85}

\begin{document}
\title{Strategies for High-Throughput FPGA-based QC-LDPC Decoder Architecture}
\author
{\IEEEauthorblockN{Swapnil Mhaske\IEEEauthorrefmark{1}, Hojin Kee\IEEEauthorrefmark{2}, Tai Ly\IEEEauthorrefmark{2}, Ahsan Aziz\IEEEauthorrefmark{2}, Predrag Spasojevic\IEEEauthorrefmark{1}}
\IEEEauthorblockA{\IEEEauthorrefmark{1}Wireless Information Network Laboratory, Rutgers University, New Brunswick, NJ, USA \\ Email:\{swapnil.mhaske@rutgers.edu, spasojev@winlab.rutgers.edu\}}
\IEEEauthorblockA{\IEEEauthorrefmark{2}National Instruments Corporation, Austin, TX, USA \\ Email:\{hojin.kee, tai.ly, ahsan.aziz\}@ni.com}}
\maketitle

\begin{abstract}
We propose without loss of generality strategies to achieve a high-throughput FPGA-based architecture for a QC-LDPC code based on a circulant-1 identity matrix construction. We present a novel representation of the parity-check matrix (PCM) providing a multi-fold throughput gain. Splitting of the node processing algorithm enables us to achieve pipelining of blocks and hence layers. By partitioning the PCM into not only layers but superlayers we derive an upper bound on the pipelining depth for the compact representation. To validate the architecture, a decoder for the \emph{IEEE 802.11n (2012)} \cite{std80211n2012} QC-LDPC is implemented on the \emph{Xilinx Kintex-7} FPGA with the help of the \emph{FPGA IP} compiler \cite{algcmp} available in the \emph{NI LabVIEW\texttrademark Communication System Design Suite (CSDS\texttrademark)} which offers an automated and systematic compilation flow where an optimized hardware implementation from the LDPC algorithm was generated in approximately 3 minutes, achieving an overall throughput of 608Mb/s (at 260MHz). As per our knowledge this is the fastest implementation of the \emph{IEEE 802.11n} QC-LDPC decoder using an algorithmic compiler.
\end{abstract}


\section{Introduction}
\label{sec:intro}
For the next generation of wireless technology collectively termed as Beyond-4G and 5G (hereafter referred to as 5G), peak data rates of upto ten Gb/s with overall latency less than $1$ms \cite{5g_nsn_cudak} are envisioned. However, due to the proposed operation in the 30-300GHz range with challenges such as short range of communication, increasing shadowing and rapid fading in time, the processing complexity of the system is expected to be high. In an effort to design and develop a channel coding solution suitable to such systems, in this paper, we present a high-throughput, scalable and reconfigurable FPGA decoder architecture. \\ 
\indent It is well known that the structure offered by QC-LDPC codes \cite{ecc_shulin} makes them amenable to time and space efficient decoder implementations relative to random LDPC codes. We believe that, given the primary requirements of high decoding throughput, QC-LDPC codes or their variants (such as accumulator-based codes \cite{cc_ryan}) that can be decoded using belief propagation (BP) methods are highly likely candidates for 5G systems. Thus, for the sole purpose of validating the proposed architecture, we chose a standard compliant code, with a throughput performance that well surpasses the requirement of the chosen standard.\\
\indent Insightful work on high-throughput (order of Gb/s) BP-based QC-LDPC decoders is available, however, most of such works focus on an ASIC design \cite{laydecarch_sun}, \cite{htqcldpcdec_zhang} which usually requires intricate customizations at the RTL level and expert knowledge of VLSI design. A sizeable subset of which caters to fully-parallel \cite{fpar_onizawa} or code-specific \cite{lcmpdec_mohsenin} architectures. From the point of view of an evolving research solution this is not an attractive option for rapid-prototyping. In the relatively less explored area of FPGA-based implementation, impressive results have recently been presented in works such as \cite{htfpga_balatsoukas},\cite{2bitmsa_chandrasetty} and \cite{mgbpsfpga_wilson}. However, these are based on fully-parallel architectures which lack flexibility (code specific) and are limited to small block sizes (primarily due to the inhibiting routing congestion) as discussed in the informative overview in \cite{gbpsdecovw_schlafer}. Since our case study is based on fully-automated generation of the HDL, a fair comparison is done with another state-of-the-art implementation \cite{vivadoldpc} in Section \ref{sec:casestudy}. Moreover, in this paper, we provide without loss of generality, strategies to achieve a high-throughput FPGA-based architecture for a QC-LDPC code based on a circulant-1 identity matrix construction. \\
\indent The main contribution of this brief is a compact representation (matrix form) of the PCM of the QC-LDPC code which provides a multi-fold increase in throughput. In spite of the resulting reduction in the degrees of freedom for pipelined processing, we achieve efficient pipelining of two-layers and also provide without loss of generality an upper bound on the pipelining depth that can be achieved in this manner. The splitting of the node processing allows us to achieve the said degree of pipelining without utilizing additional hardware resources. The algorithmic strategies were realized in hardware for our case study by the \emph{FPGA IP} \cite{algcmp} compiler in \emph{LabVIEW\texttrademark \,CSDS\texttrademark} which translated the entire software-pieplined high-level language description into VHDL in approximately 3 minutes enabling state-of-the-art rapid-prototyping. We have also demonstrated the scalability of the proposed architecture in an application that achieves over 2Gb/s of throughput \cite{ht_impl}. \\
\indent The remainder of this paper is organized as follows. Section \ref{sec:qcldpc} describes the QC-LDPC codes and the decoding algorithm chosen for this implementation. The strategies for achieving high throughput are explained in Section \ref{sec:techforht}. The case study is discussed in Section \ref{sec:casestudy}, and we conclude with Section \ref{sec:conc}.

\section{Quasi-Cyclic LDPC Codes and Decoding}
\label{sec:qcldpc}
LDPC codes (due to R. Gallager \cite{ldpc_gallager}) are a class of linear block codes that have been shown to achieve near-capacity performance on a broad range of channels and are characterized by a low-density (sparse) PCM representation. 

\begin{figure}
\centering
\includegraphics[scale=0.75]{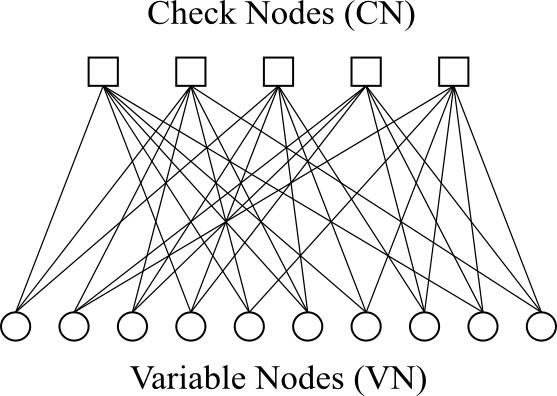}
\caption{A Tanner graph where VNs (representing the code bits) are shown as circles and CNs (representing the parity-check equations) are shown as squares. Each edge in the graph corresponds to a non-zero entry ($1$ for binary LDPC codes) in the PCM $\mathbf{H}$.}
\label{fig:tanner}
\end{figure}
Mathematically, an LDPC code is a null-space of its $m \times n$ PCM $\mathbf{H}$, where $m$ denotes the number of parity-check equations or parity-bits and $n$ denotes the number of variable nodes or code bits \cite{ecc_shulin}. In other words, for a rank $m$ PCM $\mathbf{H}$, $m$ is the number of redundant bits added to the $k$ information bits, which together form the codeword of length $n=k+m$. In the Tanner graph representation (due to Tanner \cite{ldpc_tanner}), $\mathbf{H}$ is the incidence matrix of a bipartite graph comprising of two sets: the check node (CN) set of $m$ parity-check equations and the variable node (VN) set of $n$ variable or bit nodes; the $i^{th}$ CN is connected to the $j^{th}$ VN if $\mathbf{H}(i,j)=1$, $1 \leq i \leq m$ and $1 \leq j \leq n$. A toy example of a Tanner graph is shown in Fig. \ref{fig:tanner}. The degree $d_{c_i}$ ($d_{v_j}$) of a CN $i$ (VN $j$) is equal to the number of $1$s along the $i^{th}$ row ($j^{th}$ column) of $\mathbf{H}$. For constants $c_c, c_v \in \mathbb{Z}_{>0}$ and $c_c<<m, c_v<<n$, if $\forall i,j$, $d_{c_i}=c_c$ and $d_{v_j}=c_v$, then the LDPC code is called as a regular code and is called an irregular code otherwise.

\subsection{Quasi-Cyclic LDPC Codes}
\label{subsec:qcldpc}
The first LDPC codes by Gallager are random, which complicate the decoder implementation, mainly because a random interconnect pattern between the VNs and CNs directly translates to a complex wire routing circuit on hardware. QC-LDPC codes belong to the class of structured codes that are relatively easier to implement without significantly compromising performance. \\
The construction of identity matrix based QC-LDPC codes relies on an $m_b  \times n_b$ matrix $\mathbf{H}_b$ sometimes called as the \emph{base matrix} which comprises of cyclically right-shifted identity and zero submatrices both of size $z \times z$ where, $z \in \mathbb{Z^+}, 1 \leq i_b \leq (m_b)$ and $1 \leq j_b \leq (n_b)$, the shift value,
\begin{align*}
s = \mathbf{H}_b(i_b,j_b) \in \mathcal{S} = \{-1\} \cup \{0, \ldots z-1\}
\end{align*}
The PCM matrix $\mathbf{H}$ is obtained by \emph{expanding} $\mathbf{H}_b$ using the mapping,
\begin{align*}
s \longrightarrow \left\{
  \begin{array}{lr}
  \mathbf{I}_s, &s \in \mathcal{S} \backslash \{-1\}\\
  \mathbf{0}, &s \in \{-1\}
  \end{array}
\right.
\end{align*}
where, $\mathbf{I}_s$ is an identity matrix of size $z$ which is cyclically right-shifted by $s=\mathbf{H}_b(i_b,j_b)$ and $\mathbf{0}$ is the all-zero matrix of size $z \times z$. As $\mathbf{H}$ is composed of the submatrices $\mathbf{I}_s$ and $\mathbf{0}$, it has $m=m_b.z$ rows and $n=n_b.z$ columns. 
$\mathbf{H}$ for the \emph{IEEE 802.11n (2012)} standard \cite{std80211n2012} (used for our case study) with $z=81$ is shown in Table \ref{tab:hb}.

\begin{table*}[htp]
\centering
\setlength{\tabcolsep}{0.5pt}
\begin{tabular}{l@{\hspace{0pt}}*{24}{c}}
\bfseries \bf{Layers} $\downarrow$ \quad & \multicolumn{24}{c}{\bfseries Blocks $\longrightarrow$} \\
\cmidrule(l){2-25}
&$\mathbf{B_{1}}$ &$\mathbf{B_{2}}$ &$\mathbf{B_{3}}$ &$\mathbf{B_{4}}$ &$\mathbf{B_{5}}$ &$\mathbf{B_{6}}$ &$\mathbf{B_{7}}$ &$\mathbf{B_{8}}$ &$\mathbf{B_{9}}$ &$\mathbf{B_{10}}$ &$\mathbf{B_{11}}$ &$\mathbf{B_{12}}$ &$\mathbf{B_{13}}$ &$\mathbf{B_{14}}$ &$\mathbf{B_{15}}$ &$\mathbf{B_{16}}$ &$\mathbf{B_{17}}$ &$\mathbf{B_{18}}$ &$\mathbf{B_{19}}$ &$\mathbf{B_{20}}$ &$\mathbf{B_{21}}$ &$\mathbf{B_{22}}$ &$\mathbf{B_{23}}$ &$\mathbf{B_{24}}$ \\
\midrule
\bfseries $\mathbf{L_{1}}$
&\colorbox{light-gray}{57} &-1 &-1 &-1 &\colorbox{light-gray}{50} &-1 &\colorbox{light-gray}{11} &-1 &\colorbox{light-gray}{50} &-1 &\colorbox{light-gray}{79} &-1 &\colorbox{light-gray}{1} &\colorbox{light-gray}{0} &-1 &-1 &-1 &-1 &-1 &-1 &-1 &-1 &-1 &-1\\
\midrule 
\bfseries $\mathbf{L_{2}}$
&\colorbox{light-gray}{3} &-1 &\colorbox{light-gray}{28} &-1 &\colorbox{light-gray}{0} &-1 &-1 &-1 &\colorbox{light-gray}{55} &\colorbox{light-gray}{7} &-1 &-1 &-1 &\colorbox{light-gray}{0} &\colorbox{light-gray}{0} &-1 &-1 &-1 &-1 &-1 &-1 &-1 &-1 &-1\\
\midrule
\bfseries $\mathbf{L_{3}}$
&\colorbox{light-gray}{30} &-1 &-1 &-1 &\colorbox{light-gray}{24} &\colorbox{light-gray}{37} &-1 &-1 &\colorbox{light-gray}{56} &\colorbox{light-gray}{14} &-1 &-1 &-1 &-1 &\colorbox{light-gray}{0} &\colorbox{light-gray}{0} &-1 &-1 &-1 &-1 &-1 &-1 &-1 &-1 \\
\midrule
\bfseries $\mathbf{L_{4}}$
&\colorbox{light-gray}{62} &\colorbox{light-gray}{53} &-1 &-1 &\colorbox{light-gray}{53} &-1 &-1 &\colorbox{light-gray}{3} &\colorbox{light-gray}{35} &-1 &-1 &-1 &-1 &-1 &-1 &\colorbox{light-gray}{0} &\colorbox{light-gray}{0} &-1 &-1 &-1 &-1 &-1 &-1 &-1 \\
\midrule
\bfseries $\mathbf{L_{5}}$
&\colorbox{light-gray}{40} &-1 &-1 &\colorbox{light-gray}{20} &\colorbox{light-gray}{66} &-1 &-1 &\colorbox{light-gray}{22} &\colorbox{light-gray}{28} &-1 &-1 &-1 &-1 &-1 &-1 &-1 &\colorbox{light-gray}{0} &\colorbox{light-gray}{0} &-1 &-1 &-1 &-1 &-1 &-1\\
\midrule
\bfseries $\mathbf{L_{6}}$
&\colorbox{light-gray}{0} &-1 &-1 &-1 &\colorbox{light-gray}{8} &-1 &\colorbox{light-gray}{42} &-1 &\colorbox{light-gray}{50} &-1 &-1 &\colorbox{light-gray}{8} &-1 &-1 &-1 &-1 &-1 &\colorbox{light-gray}{0} &\colorbox{light-gray}{0} &-1 &-1 &-1 &-1 &-1\\
\midrule
\bfseries $\mathbf{L_{7}}$
&\colorbox{light-gray}{69} &\colorbox{light-gray}{79} &\colorbox{light-gray}{79} &-1 &-1 &-1 &\colorbox{light-gray}{56} &-1 &\colorbox{light-gray}{52} &-1 &-1 &-1 &\colorbox{light-gray}{0} &-1 &-1 &-1 &-1 &-1 &\colorbox{light-gray}{0} &\colorbox{light-gray}{0} &-1 &-1 &-1 &-1\\
\midrule
\bfseries $\mathbf{L_{8}}$
&\colorbox{light-gray}{65} &-1 &-1 &-1 &\colorbox{light-gray}{38} &\colorbox{light-gray}{57} &-1 &-1 &\colorbox{light-gray}{72} &-1 &\colorbox{light-gray}{27} &-1 &-1 &-1 &-1 &-1 &-1 &-1 &-1 &\colorbox{light-gray}{0} &\colorbox{light-gray}{0} &-1 &-1 &-1\\
\midrule
\bfseries $\mathbf{L_{9}}$
&\colorbox{light-gray}{64} &-1 &-1 &-1 &\colorbox{light-gray}{14} &\colorbox{light-gray}{52} &-1 &-1 &\colorbox{light-gray}{30} &-1 &-1 &\colorbox{light-gray}{32} &-1 &-1 &-1 &-1 &-1 &-1 &-1 &-1 &\colorbox{light-gray}{0} &\colorbox{light-gray}{0} &-1 &-1\\
\midrule
\bfseries $\mathbf{L_{10}}$
&-1 &\colorbox{light-gray}{45} &-1 &\colorbox{light-gray}{70} &\colorbox{light-gray}{0} &-1 &-1 &-1 &\colorbox{light-gray}{77} &\colorbox{light-gray}{9} &-1 &-1 &-1 &-1 &-1 &-1 &-1 &-1 &-1 &-1 &-1 &\colorbox{light-gray}{0} &\colorbox{light-gray}{0} &-1\\
\midrule
\bfseries $\mathbf{L_{11}}$
&\colorbox{light-gray}{2} &\colorbox{light-gray}{56} &-1 &\colorbox{light-gray}{57} &\colorbox{light-gray}{35} &-1 &-1 &-1 &-1 &-1 &\colorbox{light-gray}{12} &-1 &-1 &-1 &-1 &-1 &-1 &-1 &-1 &-1 &-1 &-1 &\colorbox{light-gray}{0} &\colorbox{light-gray}{0}\\
\midrule
\bfseries $\mathbf{L_{12}}$
&\colorbox{light-gray}{24} &-1 &\colorbox{light-gray}{61} &-1 &\colorbox{light-gray}{60} &-1 &-1 &\colorbox{light-gray}{27} &\colorbox{light-gray}{51} &-1 &-1 &\colorbox{light-gray}{16} &\colorbox{light-gray}{1} &-1 &-1 &-1 &-1 &-1 &-1 &-1 &-1 &-1 &-1 &\colorbox{light-gray}{0} \\
\bottomrule
\addlinespace
\end{tabular}

\caption{Base matrix $\mathbf{H}_b$ for $z=81$ specified in IEEE 802.11n (2012) standard used in the case study. $L_1-L_{12}$ are the layers and $B_1-B_{24}$ are the block columns (see Section \ref{subsec:layered}). Valid blocks (see section \ref{subsec:beta}) are highlighted.}
\label{tab:hb}
\end{table*}

\subsection{Scaled Min-Sum Algorithm for Decoding QC-LDPC Codes}
\label{subsec:msa}
LDPC codes can be decoded using message passing (MP) or belief propagation (BP) \cite{ldpc_gallager,spa_factorgraphs} on the bipartite Tanner graph where, the CNs and VNs communicate with each other, successively passing revised estimates of the log-likelihood ratio (LLR) associated, in every decoding iteration. In this work we have employed the efficient decoding algorithm presented in \cite{serialmp_litsyn}, with pipelined processing of layers based on the row-layered decoding technique \cite{laydec}, detailed in Section \ref{subsec:layered}.
\begin{defin}
For $1 \leq i \leq m$ and $1 \leq j \leq n$, let $v_j$ denote the $j^{th}$ bit in the length $n$ codeword and $y_j=v_j+n_j$ denote the corresponding received value from the channel corrupted by the noise sample $n_j$. Let the variable-to-check (VTC) message from VN $j$ to CN $i$ be $q_{ij}$ and, let the check-to-variable (CTV) message from CN $i$ to VN $j$ be $r_{ij}$. Let the a posteriori probability (APP) ratio for VN $j$ be denoted as $p_j$.
\end{defin}
\noindent The steps of the scaled-MSA \cite{ecc_shulin},\cite{smsa_chen} are given below.
\begin{enumerate}
\item Initialize the APP ratio and the CTV messages as, 
\begin{align} \label{eq:initllr}
p_j^{(0)} &= ln \left\lbrace \frac{P(v_j=0|y_j)}{P(v_j=1|y_j)} \right\rbrace, \quad 1 \leq j \leq n \\
r_{ij}^{(0)} &= 0, \quad 1 \leq i \leq m, 1 \leq j \leq n \nonumber
\end{align}
\item Iteratively compute at the $t^{th}$ decoding iteration,
\begin{align} \label{eq:vnmsg}
q_{ij}^{(t)} &= p_{j}^{(t-1)}-r_{ij}^{(t-1)}
\end{align}
\begin{align} \label{eq:cnmsg}
r_{ij}^{(t)} &= a \cdot \prod_{k \in \mathcal{N}(i)\backslash\{j\}} sign \left(q_{ik}^{(t)} \right) \cdot \min_{k \in \mathcal{N}(i)\backslash\{j\}} \left\lbrace |q_{ik}^{(t)}| \right\rbrace
\end{align}
\begin{align} \label{eq:app}
p_j^{(t)} &= q_{ij}^{(t)} + r_{ij}^{(t)}
\end{align}
where, $1 \leq i \leq m$, and $k \in \mathcal{N}(i)\backslash\{j\}$ represents the set of the VN neighbors of CN $i$ excluding VN $j$. Let $t_{max}$ be the maximum number of decoding iterations.
\item Decision on the code bit $v_j$, $1 \leq j \leq n$ as,
\begin{align}
\hat{v}_j = \left\{
  \begin{array}{lr}
  0, &p_j < 0 \\
  1, &otherwise
  \end{array}
\right.
\end{align}
\item If $\hat{\mathbf{v}}\mathbf{H}^{T}=0$, where $\hat{\mathbf{v}}=(\hat{v}_1,\hat{v}_2,\ldots,\hat{v}_n)$, or $t=t_{max}$, declare $\hat{\mathbf{v}}$ as the decoded codeword.
\end{enumerate}

It is well known that since the MSA is an approximation of the SPA, the performance of the MSA is relatively worse than the SPA \cite{ecc_shulin}. However, in \cite{smsa_chen} it has been shown that scaling the CTV messages $r_{ij}$ can improve the performance of the MSA. Hence, we scale the CTV messages by a factor $a$ (=$0.75$).
\begin{remark} \label{rem:algoadv}
\indent The standard MP algorithm is based on the so-called \emph{flooding} or \emph{two-phase} schedule where, each decoding iteration comprises of two phases. In the first phase, VTC messages for all the VNs are computed and, in the second phase the CTV messages for all the CNs are computed, strictly in that order. Thus, message updates from one side of the graph propagate to the other side only in the next decoding iteration. In the algorithm given in \cite{serialmp_litsyn} however, message updates can propagate across the graph in the same decoding iteration. This provides advantages \cite{serialmp_litsyn} such as, a single processing unit is required for both CN and VN message updates, memory storage is reduced on account of the on-the-fly computation of the VTC messages $q_{ij}$ and the algorithm converges faster than the standard MP flooding schedule requiring fewer decoding iterations.
\end{remark}

\section{Techniques for High-throughput}
\label{sec:techforht}
To understand the high-throughput requirements for LDPC decoding, let us first define the decoding throughput $T$ of an iterative LDPC decoder.
\begin{defin} \label{def:thruput}
Let $F_{c}$ be the clock frequency, $n$ be the code length, $N_{i}$ be the number of decoding iterations and $N_{c}$ be the number of clock cycles per decoding iteration, then the throughput of the decoder is given by, $T=\frac{F_{c} \cdot n}{N_{i} \cdot N_{c}} \quad \text{b/s}$
\end{defin}
Even though, $n$ and $N_i$ are functions of the code and the decoding algorithm used, $F_c$ and $N_c$ are determined by the hardware architecture. Architectural optimization such as the ability to operate the decoder at higher clock rates with minimal latency between decoding iterations can help achieve higher throughput. We have employed the following techniques to increase the throughput given by Definition \ref{def:thruput}.

\subsection{Linear Complexity Node Processing}
\label{subsec:npu}
As noted in Section \ref{subsec:msa}, separate processing units for CNs and VNs are not required unlike that for the flooding schedule. The hardware elements that process equations (\ref{eq:vnmsg})-(\ref{eq:app}) are collectively referred to as the Node Processing Unit (NPU). \\
\indent Careful observation reveals that, among equations (\ref{eq:vnmsg})-(\ref{eq:app}), processing the CTV messages $r_{ij}$, $1 \leq i \leq m$ and $1 \leq j \leq n$ is the most computationally intensive due to the calculation of the sign, and the minimum value operations. The complexity of processing the minimum value is $\mathcal{O}(d_{c_i}^2)$.
In software, this translates to two nested for-loops, an outer loop that executes $d_{c_i}$ times and an inner loop that executes $(d_{c_i}-1)$ times. \\
\indent To achieve linear complexity $\mathcal{O}(d_{c_i})$ for the minimum value computation, we split the process into two phases or passes: the \emph{global} pass where the first and the second minimum (the smallest value in the set excluding the minimum value of the set) for all the neighboring VNs of a CN are computed and the \emph{local} pass where the first and second minimum from the global pass are used to compute the minimum value for each neighboring VN.  Based on the functionality of the two passes, the NPU is divided into the Global NPU (GNPU) and the Local NPU (LNPU). The algorithm is given below.
\begin{enumerate}
\item \emph{Global Pass}:
\begin{enumerate}
\item[i.] \emph{Initialization}: Let $\ell$ denote the discrete time-steps such that, $\ell \in \{0\}\cup \{1,2,\ldots,|\mathcal{N}(i)|\}$ and let $f^{(\ell)}$ and $s^{(\ell)}$ denote the value of the first and the second minimum at time $\ell$ respectively. The initial value at time $\ell=0$ is,
\begin{align}
f^{(0)}=s^{(0)}=\infty.
\end{align}

\item[ii.] \emph{Comparison}: Let $k_i(\ell) \in \mathcal{N}(i)$, $\ell = \{1,2,\ldots,|\mathcal{N}(i)|\}$, denote the index of the $\ell^{th}$ neighboring VN of CN $i$. Note that, $k_i(\ell)$ depends on $i$ and $\ell$, specifically, for a given CN $i$ it is a bijective function of $\ell$. An increment from $(\ell-1)$ to $\ell$ corresponds to moving from the edge CN $i$ $\leftrightarrow$ VN $k_i(\ell-1)$ to the edge CN $i$ $\leftrightarrow$ VN $k_i(\ell)$. \\
\begin{align}
f^{(\ell)} &= \left\{
\begin{array}{ll}
      |q_{ik_i(\ell)}|,	& |q_{ik_i(\ell)}| \leq f^{(\ell-1)} \\
      f^{(\ell-1)}, 		& otherwise. \\
\end{array}
\right. \\
s^{(\ell)} &= \left\{
\begin{array}{ll}
      |q_{ik_i(\ell)}|,	& f^{(\ell-1)} < |q_{ik_i(\ell)}| < s^{(\ell-1)} \\
      f^{(\ell-1)}, 		& |q_{ik_i(\ell)}| \leq f^{(\ell-1)} \\
      s^{(\ell-1)},		& otherwise.
\end{array}
\right.
\end{align}
\end{enumerate}
Thus, $f^{(\ell_{max})}$ and $s^{(\ell_{max})}$ are the first and second minimum values for the set of VN neighbors of CN $i$, where, $\ell_{max}=|\mathcal{N}(i)|$.
\item \emph{Local Pass}: Let the minimum value as per equation (\ref{eq:cnmsg}) for VN $k_i(\ell)$ be denoted as $q^{min}_{i k_i(\ell)}$, $\ell \in \{1,2,\ldots,|\mathcal{N}(i)|\}$ then,
\begin{align}
q^{min}_{i k_i(\ell)} = \left\{
\begin{array}{ll}
      f^{(\ell_{max})},	& |q_{ik_i(\ell)}| \neq f^{(\ell_{max})} \\
      s^{(\ell_{max})}, 	& otherwise. \\
\end{array}
\right.
\end{align}
\end{enumerate}
In software, this translates to two consecutive for-loops, each executing $(d_{c_i}-1)$ times. Consequently, this reduces the complexity from $\mathcal{O}(d_{c_i}^2)$ to $\mathcal{O}(d_{c_i})$. A similar approach is also found in \cite{valreuse_gunnam}, \cite{laydecarch_sun}. The sign computation is processed in a similar manner.

\begin{table}[ht]
\centering
\setlength{\tabcolsep}{0pt}
\scalebox{0.88}{
\begin{tabular}{l@{\hspace{2pt}}*{8}{c}}
\bfseries  
&$\mathbf{VN_{z\mathbf{J}}}$ &\ldots &$\textcolor{white}{\mathbf{VN_{z(\mathbf{J}-1)}}}$ &$\mathbf{VN_{z\mathbf{J}+l}}$ &$\textcolor{white}{\mathbf{VN_{z(\mathbf{J}+1)-1}}}$ &\ldots &$\mathbf{VN_{z(\mathbf{J}+1)-1}}$ \\
\cmidrule(r){2-8}
\bfseries $\mathbf{NPU_{0}}$
&0 &\ldots &0 &1 &0 &\ldots &0 \\
\cmidrule(r){2-8}
\bfseries $\mathbf{NPU_{1}}$
&0 &\ldots &0 &0 &1 &\ldots &0 \\
\cmidrule(r){2-8}
\bfseries \quad \quad \vdots 
&\vdots & & & & & &\vdots \\
\cmidrule(r){2-8}
\bfseries $\mathbf{NPU_{z-2}}$
&0 &\ldots &0 &0 &0 &\ldots &0 \\
\cmidrule(r){2-8}
\bfseries $\mathbf{NPU_{z-1}}$
&0 &\ldots &1 &0 &0 &\ldots &0 \\
\cmidrule(r){2-8}
\addlinespace \addlinespace
\end{tabular}
}
\caption{Arbitrary submatrix $\mathbf{I}_s$ in $\mathbf{H}$, $0 \leq J \leq n_b-1$, illustrating the opportunity to parallelize $z$ NPUs.}
\label{tab:zunroll}
\end{table}

\subsection{$z$-fold Parallelization of NPUs}
\label{subsec:degzpar}
The CN message computation given by equation (\ref{eq:cnmsg}) is repeated $m$ times in a decoding iteration i.e. once for each CN. A straightforward serial implementation of this kind is slow and undesirable. Instead, we apply a strategy based on the following understanding.
\begin{fact} \label{fac:znpu}
An arbitrary submatrix $\mathbf{I}_s$ in the PCM $\mathbf{H}$ corresponds to $z$ CNs connected to $z$ VNs on the bipartite graph, with strictly $1$ edge between each CN and VN.
\end{fact}
This implies that no CN in this set of $z$ CNs given by $\mathbf{I}_s$ shares a VN with another CN in the same set. Table \ref{tab:zunroll} illustrates such an arbitrary submatrix in $\mathbf{H}$. This presents us with an opportunity to operate $z$ NPUs in parallel (hereafter referred to as an \emph{NPU array}), resulting in a $z$-fold increase in throughput.

\subsection{Layered Decoding}
\label{subsec:layered}
From Remark \ref{rem:algoadv} it is clear that, in the flooding schedule \emph{all} nodes on one side of the bipartite graph can be processed in parallel. Although, such a \emph{fully parallel} implementation may seem as an attractive option for achieving high-throughput performance, it has its own drawbacks. Firstly, it becomes quickly intractable in hardware due to the complex interconnect pattern. Secondly, such an implementation usually restricts itself to a specific code structure. In spite of the 
serial nature of the algorithm in \ref{subsec:msa}, one can process multiple nodes at the same time if the following condition is satisfied.
\begin{fact} \label{fac:depcond}
From the perspective of CN processing, two or more CNs can be processed at the same time (i.e. they are independent of each other) if they do not have one or more VNs (code bits) in common.
\end{fact}
The row-layering technique used in this work essentially relies on the above condition being satisfied. In terms of $\mathbf{H}$, an arbitrary subset of rows can be processed at the same time provided that, no two or more rows have a $1$ in the same column of $\mathbf{H}$. This subset of rows is termed as a \emph{row-layer} (hereafter referred to as a \emph{layer}). In other words, given a set $\mathcal{L}=\{L_1, L_2, \ldots,L_I\}$ of $I$ layers in $\mathbf{H}$, $\forall u \in \{1,2,\ldots,I\}$ and $\forall i, i^\prime \in L_u$, then, $\mathcal{N}(i) \cap \mathcal{N}(i^\prime)=\phi.$ \\
\indent Observing that, $\sum_{u=1}^{I}|L_u|=m$, in general, $L_u$ can be any subset of rows as long as the rows within each subset satisfy the condition in Fact \ref{fac:depcond}; implying that, $|L_u| \neq |L_{u^\prime}|$, $\forall u, u^\prime \in \{1,2,\ldots,I\}$ is possible.
Owing to the structure of QC-LDPC codes, the choice of $|L_u|$ (and hence $I$) becomes much obvious. Submatrices $\mathbf{I}_s$ in $\mathbf{H}_b$ (with row and column weight of $1$) guarantee that, for the $z$ CNs corresponding to the rows of $\mathbf{I}_s$), always satisfy the condition in Fact \ref{fac:depcond}. Hence, $|L_u|=|L_{u^\prime}|=z$ is chosen. \\
\indent From the VN or column perspective, $|L_u|=z$, $\forall u = \{1,2,\ldots,I\}$ implies that, the columns of $\mathbf{H}$ are also divided into subsets of size $z$ (hereafter referred to as \emph{block columns}) given by the set $\mathcal{B}=\{B_1,B_2,\ldots,B_J\}$, $J=\frac{n}{z}=n_b$.
Observing that VNs belonging to a block column may participate in CN equations across several layers, we further divide the block columns into \emph{blocks}, where a block is the intersection of a layer and a block column. Two or more layers $L_u,L_{u^\prime}$  are said to be \emph{dependent} with respect to the block column $B_w$ if, $\mathbf{H}_b(u,w) \neq -1$ and, $\mathbf{H}_b(u^\prime,w) \neq -1$ and are said to be \emph{independent} otherwise.

\begin{table}[htp]
\centering
\setlength{\tabcolsep}{6pt}
\begin{tabular}{l@{\hspace{4pt}}*{5}{c}}
\bfseries \bf{Layers} $\downarrow$ & \multicolumn{5}{c}{\bfseries Blocks $\longrightarrow$} \\
\cmidrule(l){1-6}
&\ldots &$\mathbf{B_{2}}$ &$\mathbf{B_{3}}$ &$\mathbf{B_{4}}$ &\ldots \\
\midrule
\bfseries $\mathbf{L_{1}}$
&\ldots &$\downarrow$ &$\downarrow$ &$\downarrow$ &\ldots \\
\midrule 
\bfseries $\mathbf{L_{2}}$
&\ldots &$\downarrow$ &\colorbox{light-gray}{28} &$\downarrow$ &\ldots \\
\midrule
\bfseries $\mathbf{L_{3}}$
&\ldots &$\downarrow$ &$\downarrow$ &$\downarrow$ &\ldots \\
\midrule
\bfseries $\mathbf{L_{4}}$
&\ldots &\colorbox{light-gray}{53} &$\downarrow$ &$\downarrow$ &\ldots \\
\midrule
\bfseries $\mathbf{L_{5}}$
&\ldots &$\downarrow$ &$\downarrow$ &\colorbox{light-gray}{20} &\ldots \\
\midrule
\bfseries $\mathbf{L_{6}}$
&\ldots &$\downarrow$ &$\downarrow$ &$\downarrow$ &\ldots \\
\midrule
\bfseries $\mathbf{L_{7}}$
&\ldots &\colorbox{light-gray}{79} &\colorbox{light-gray}{79} &$\downarrow$ &\ldots \\
\midrule
\bfseries $\mathbf{L_{8}}$
&\ldots &$\downarrow$ &$\downarrow$ &$\downarrow$ &\ldots \\
\midrule
\bfseries $\mathbf{L_{9}}$
&\ldots &$\downarrow$ &$\downarrow$ &$\downarrow$ &\ldots \\
\midrule
\bfseries $\mathbf{L_{10}}$
&\ldots &\colorbox{light-gray}{45} &$\downarrow$ &\colorbox{light-gray}{70} &\ldots \\
\midrule
\bfseries $\mathbf{L_{11}}$
&\ldots &\colorbox{light-gray}{56} &$\downarrow$ &\colorbox{light-gray}{57} &\ldots \\
\midrule
\bfseries $\mathbf{L_{12}}$
&\ldots &$\downarrow$ &\colorbox{light-gray}{61} &$\downarrow$ &\ldots \\
\midrule
\bfseries 
& &to $L_4$ &to $L_2$ &to $L_5$ & \\
\addlinespace \addlinespace
\end{tabular}
\caption{Illustration of Message Passing in row-layered decoding in a Section of the PCM $\mathbf{H}_b$.}
\label{tab:betadep}
\end{table}

For example, in Table \ref{tab:betadep} we can see that layers $L_4,L_7,L_{10}$ and $L_{11}$ are dependent with respect to block column $B_2$. Assuming that the message update begins with layer $L_1$ and proceeds downward, the arrows represent the directional flow of message updates from one layer to another. Thus, layer $L_7$ cannot begin updating the VNs associated with block column $B_2$ before layer $L_4$ has finished updating messages for the same set of VNs and so on. The idea of parallelizing z NPUs seen in Section \ref{subsec:degzpar} can be extended to layers, NPU arrays can process message updates for multiple independent layers. It is clear that, dependent layers limit the degree of parallelization available to achieve high-throughput. In Section \ref{subsec:ppldecarch}, we discuss pipelining methods that allow us to overcome layer-to-layer dependency and improve throughput.

\subsection{Compact Representation of $\mathbf{H}_b$}
\label{subsec:beta}
Before we discuss the pipelined processing of layers, we present a novel compact (thus efficient) matrix representation leading to a significant improvement in throughput. To understand this, let us call $\mathbf{0}$ submatrices in $\mathbf{H}$ as \emph{invalid} blocks, where there are no edges between the corresponding CNs and VNs, and the submatrices $\mathbf{I}_s$ as \emph{valid} blocks. In a conventional approach to scheduling (for example in \cite{htqcldpcdec_zhang}), message computation is done for all the valid and invalid blocks. To avoid processing invalid blocks, we propose an alternate representation of $\mathbf{H}_b$ in the form of two matrices: $\boldsymbol\beta_I$ (Table \ref{tab:betai}), the block index matrix and $\boldsymbol\beta_S$ (Table \ref{tab:betas}), the block shift matrix.
\begin{table}[htp] 
\centering
\setlength{\tabcolsep}{2.5pt}
\scalebox{1}{
\begin{tabular}{l@{\hspace{4pt}}*{8}{c}}
\bfseries \bf{Layers} $\downarrow$ \quad & \multicolumn{8}{c}{\bfseries Blocks $\longrightarrow$} \\
\cmidrule(l){2-9}
&$\mathbf{b_{1}}$ &$\mathbf{b_{2}}$ &$\mathbf{b_{3}}$ &$\mathbf{b_{4}}$ &$\mathbf{b_{5}}$ &$\mathbf{b_{6}}$ &$\mathbf{b_{7}}$ &$\mathbf{b_{8}}$ \\
\midrule
\bfseries $\mathbf{L_{1}}$
&\colorbox{light-gray}{0} &\colorbox{light-gray}{4} &\colorbox{light-gray}{6} &\colorbox{light-gray}{8} &\colorbox{light-gray}{10} &\colorbox{light-gray}{12} &\colorbox{light-gray}{13} &-1 \\
\midrule 
\bfseries $\mathbf{L_{2}}$
&\colorbox{light-gray}{0} &\colorbox{light-gray}{2} &\colorbox{light-gray}{4} &\colorbox{light-gray}{8} &\colorbox{light-gray}{9} &\colorbox{light-gray}{13} &\colorbox{light-gray}{14} &-1 \\
\midrule
\bfseries $\mathbf{L_{3}}$
&\colorbox{light-gray}{0} &\colorbox{light-gray}{4} &\colorbox{light-gray}{5} &\colorbox{light-gray}{8} &\colorbox{light-gray}{9} &\colorbox{light-gray}{14} &\colorbox{light-gray}{15} &-1 \\
\midrule
\bfseries $\mathbf{L_{4}}$
&\colorbox{light-gray}{0} &\colorbox{light-gray}{1} &\colorbox{light-gray}{4} &\colorbox{light-gray}{7} &\colorbox{light-gray}{8} &\colorbox{light-gray}{15} &\colorbox{light-gray}{16} &-1 \\
\midrule
\bfseries $\mathbf{L_{5}}$
&\colorbox{light-gray}{0} &\colorbox{light-gray}{3} &\colorbox{light-gray}{4} &\colorbox{light-gray}{7} &\colorbox{light-gray}{8} &\colorbox{light-gray}{16} &\colorbox{light-gray}{17} &-1 \\
\midrule
\bfseries $\mathbf{L_{6}}$
&\colorbox{light-gray}{0} &\colorbox{light-gray}{4} &\colorbox{light-gray}{6} &\colorbox{light-gray}{8} &\colorbox{light-gray}{11} &\colorbox{light-gray}{17} &\colorbox{light-gray}{18} &-1 \\
\midrule
\bfseries $\mathbf{L_{7}}$
&\colorbox{light-gray}{0} &\colorbox{light-gray}{1} &\colorbox{light-gray}{2} &\colorbox{light-gray}{6} &\colorbox{light-gray}{8} &\colorbox{light-gray}{12} &\colorbox{light-gray}{18} &\colorbox{light-gray}{19} \\
\midrule
\bfseries $\mathbf{L_{8}}$
&\colorbox{light-gray}{0} &\colorbox{light-gray}{4} &\colorbox{light-gray}{5} &\colorbox{light-gray}{8} &\colorbox{light-gray}{10} &\colorbox{light-gray}{19} &\colorbox{light-gray}{20} &-1 \\
\midrule
\bfseries $\mathbf{L_{9}}$
&\colorbox{light-gray}{0} &\colorbox{light-gray}{4} &\colorbox{light-gray}{5} &\colorbox{light-gray}{8} &\colorbox{light-gray}{11} &\colorbox{light-gray}{20} &\colorbox{light-gray}{21} &-1 \\
\midrule
\bfseries $\mathbf{L_{10}}$
&\colorbox{light-gray}{1} &\colorbox{light-gray}{3} &\colorbox{light-gray}{4} &\colorbox{light-gray}{8} &\colorbox{light-gray}{9} &\colorbox{light-gray}{21} &\colorbox{light-gray}{22} &-1 \\
\midrule
\bfseries $\mathbf{L_{11}}$
&\colorbox{light-gray}{0} &\colorbox{light-gray}{1} &\colorbox{light-gray}{3} &\colorbox{light-gray}{4} &\colorbox{light-gray}{10} &\colorbox{light-gray}{22} &\colorbox{light-gray}{23} &-1 \\
\midrule
\bfseries $\mathbf{L_{12}}$
&\colorbox{light-gray}{0} &\colorbox{light-gray}{2} &\colorbox{light-gray}{4} &\colorbox{light-gray}{7} &\colorbox{light-gray}{8} &\colorbox{light-gray}{11} &\colorbox{light-gray}{12} &\colorbox{light-gray}{23} \\
\bottomrule
\addlinespace
\end{tabular}
}
\caption{Block index matrix $\boldsymbol{\beta_{I}}$ showing the valid blocks (highlighted) to be processed.}
\label{tab:betai}
\end{table}
$\boldsymbol\beta_I$ and $\boldsymbol\beta_S$ hold the index locations and the shift values (and hence the connections between the CNs and VNs) corresponding to \emph{only} the valid blocks in $\mathbf{H}_b$, respectively. Construction of $\boldsymbol\beta_I$ is based on the following definition,
\begin{defin} \label{def:betai}
Construction of $\boldsymbol\beta_I$ is as follows. \\
for $u = \{1,2,\ldots,I\}$ \\
\indent set $w=0$, $j_b=0$ \\
\indent for $j_b=\{1,2,\ldots,n_b\}$ \\ 
\indent \indent $j_b = j_b+1$ \\
\indent \indent if $\mathbf{H}_b(u,j_b) \neq -1$ \\
\indent \indent \indent $w=w+1; \boldsymbol\beta_I(u,w)=j_b; \boldsymbol\beta_S(u,w)=\mathbf{H}_b(u,j_b).$
\end{defin}
To observe the benefit of this alternate representation, let us define the following ratio.
\begin{defin}
Let $\lambda$ denote the compaction ratio, which is the ratio of the number of columns of $\boldsymbol{\beta}_I$ (which is the same for $\boldsymbol{\beta}_S$) to the number of columns of $\mathbf{H}_b$. Hence, $\lambda = \frac{J}{n_b}$. 
\end{defin}
The compaction ratio $\lambda$ is a measure of the compaction achieved by the alternate representation of $H_b$.
Compared to the conventional approach, scheduling as per the $\boldsymbol\beta_I$ and $\boldsymbol\beta_S$ matrices improves throughput by $\frac{1}{\lambda}$ times. In our case study, $\lambda = \frac{8}{24} = \frac{1}{3}$, thus providing a throughput gain of $\frac{1}{\lambda}=3$.

\begin{table}[htp] 
\centering
\setlength{\tabcolsep}{2.5pt}
\begin{tabular}{l@{\hspace{4pt}}*{8}{c}}
\bfseries \bf{Layers} $\downarrow$ \quad & \multicolumn{8}{c}{\bfseries Blocks $\longrightarrow$} \\
\cmidrule(l){2-9}
&$\mathbf{b_{1}}$ &$\mathbf{b_{2}}$ &$\mathbf{b_{3}}$ &$\mathbf{b_{4}}$ &$\mathbf{b_{5}}$ &$\mathbf{b_{6}}$ &$\mathbf{b_{7}}$ &$\mathbf{b_{8}}$ \\
\midrule
\bfseries $\mathbf{L_{1}}$
&\colorbox{light-gray}{57} &\colorbox{light-gray}{50} &\colorbox{light-gray}{11} &\colorbox{light-gray}{50} &\colorbox{light-gray}{79} &\colorbox{light-gray}{1} &\colorbox{light-gray}{0} &-1 \\
\midrule 
\bfseries $\mathbf{L_{2}}$
&\colorbox{light-gray}{3} &\colorbox{light-gray}{28} &\colorbox{light-gray}{0} &\colorbox{light-gray}{55} &\colorbox{light-gray}{7} &\colorbox{light-gray}{0} &\colorbox{light-gray}{0} &-1 \\
\midrule
\bfseries $\mathbf{L_{3}}$
&\colorbox{light-gray}{30} &\colorbox{light-gray}{24} &\colorbox{light-gray}{37} &\colorbox{light-gray}{56} &\colorbox{light-gray}{14} &\colorbox{light-gray}{0} &\colorbox{light-gray}{0} &-1 \\
\midrule
\bfseries $\mathbf{L_{4}}$
&\colorbox{light-gray}{62} &\colorbox{light-gray}{53} &\colorbox{light-gray}{53} &\colorbox{light-gray}{3} &\colorbox{light-gray}{35} &\colorbox{light-gray}{0} &\colorbox{light-gray}{0} &-1 \\
\midrule
\bfseries $\mathbf{L_{5}}$
&\colorbox{light-gray}{40} &\colorbox{light-gray}{20} &\colorbox{light-gray}{66} &\colorbox{light-gray}{22} &\colorbox{light-gray}{28} &\colorbox{light-gray}{0} &\colorbox{light-gray}{0} &-1 \\
\midrule
\bfseries $\mathbf{L_{6}}$
&\colorbox{light-gray}{0} &\colorbox{light-gray}{8} &\colorbox{light-gray}{42} &\colorbox{light-gray}{50} &\colorbox{light-gray}{8} &\colorbox{light-gray}{0} &\colorbox{light-gray}{0} &-1 \\
\midrule
\bfseries $\mathbf{L_{7}}$
&\colorbox{light-gray}{69} &\colorbox{light-gray}{79} &\colorbox{light-gray}{79} &\colorbox{light-gray}{56} &\colorbox{light-gray}{52} &\colorbox{light-gray}{0} &\colorbox{light-gray}{0} &\colorbox{light-gray}{0} \\
\midrule
\bfseries $\mathbf{L_{8}}$
&\colorbox{light-gray}{65} &\colorbox{light-gray}{38} &\colorbox{light-gray}{57} &\colorbox{light-gray}{72} &\colorbox{light-gray}{27} &\colorbox{light-gray}{0} &\colorbox{light-gray}{0} &-1 \\
\midrule
\bfseries $\mathbf{L_{9}}$
&\colorbox{light-gray}{64} &\colorbox{light-gray}{14} &\colorbox{light-gray}{52} &\colorbox{light-gray}{30} &\colorbox{light-gray}{32} &\colorbox{light-gray}{0} &\colorbox{light-gray}{0} &-1 \\
\midrule
\bfseries $\mathbf{L_{10}}$
&\colorbox{light-gray}{45} &\colorbox{light-gray}{70} &\colorbox{light-gray}{0} &\colorbox{light-gray}{77} &\colorbox{light-gray}{9} &\colorbox{light-gray}{0} &\colorbox{light-gray}{0} &-1 \\
\midrule
\bfseries $\mathbf{L_{11}}$
&\colorbox{light-gray}{2} &\colorbox{light-gray}{56} &\colorbox{light-gray}{57} &\colorbox{light-gray}{35} &\colorbox{light-gray}{12} &\colorbox{light-gray}{0} &\colorbox{light-gray}{0} &-1 \\
\midrule
\bfseries $\mathbf{L_{12}}$
&\colorbox{light-gray}{24} &\colorbox{light-gray}{61} &\colorbox{light-gray}{60} &\colorbox{light-gray}{27} &\colorbox{light-gray}{51} &\colorbox{light-gray}{16} &\colorbox{light-gray}{1} &\colorbox{light-gray}{0} \\
\bottomrule
\addlinespace
\end{tabular}
\caption{Block shift matrix $\boldsymbol{\beta_{S}}$ showing the right-shift values for the valid blocks to be processed.}
\label{tab:betas}
\end{table}

\begin{remark}
In the irregular QC-LDPC code in our case study, all layers comprise of $7$ blocks each, except layer $L_7$ and $L_{12}$ which have $8$. With the aim of minimizing hardware complexity by maintaining a static memory-address generation pattern (does not change from layer-to-layer), our implementation assumes regularity in the code. The decoder processes $8$ blocks for each layer of the $\boldsymbol\beta_I$ matrix resulting in some throughput penalty. The results from processing the invalid blocks in $L_7$ and $L_{12}$ are not stored in the memory.
\end{remark}

\subsection{Layer-Pipelined Decoder Architecture}
\label{subsec:ppldecarch}
In Section \ref{subsec:layered} we saw how dependent layers for a block column cannot be processed in parallel. For instance, in $\mathbf{H}_b$ in Table \ref{tab:hb}, VNs associated with the block column $B_1$ participate in CN equations associated with all the layers except layer $L_{10}$, suggesting that there is no scope of parallelization of layer processing at all. This situation is better observed in $\boldsymbol{\beta_I}$ shown in Table \ref{tab:betai}. 
\begin{fact} \label{fac:betadep}
\emph{If a block column of $\boldsymbol{\beta_I}$ has a particular index value appearing in more than one layer, then the layers corresponding to that value are dependent with respect to that block column}.
\end{fact}
\begin{proof}
Follows directly by applying Fact \ref{fac:depcond} to Definition \ref{def:betai}.
\end{proof}

In other words, $\forall u,u^\prime \in \{1,2,\ldots,I\}$, $\forall w \in \{1,2,\ldots,J\}$, if,
$\boldsymbol{\beta_I}(u,w) = \boldsymbol{\beta_I}(u^\prime,w)$
then, the layers $L_u$ and $L_{u^\prime}$ are dependent.
It is obvious that, to process all layers in parallel ($L_{1}$ to $L_{12}$ in \ref{tab:hb}), the condition,
\begin{align} \label{eq:indlaycond}
\boldsymbol{\beta_I}(u,w) \neq \boldsymbol{\beta_I}(u^\prime,w)
\end{align}
must hold for $\forall u,u^\prime \in \{1,2,\ldots,I\}$. 
We call the set of layers $\mathcal{L}$ satisfying Fact \ref{fac:betadep} as a \emph{superlayer}. As will be seen later, the formation of superlayers of suitable size is crucial to achieve parallelism in the architecture.

\begin{table}[htp] 
\centering
\setlength{\tabcolsep}{2.5pt}
\scalebox{1}{
\begin{tabular}{l@{\hspace{4pt}}*{8}{c}}
\bfseries \bf{Layers} $\downarrow$ \quad & \multicolumn{8}{c}{\bfseries Blocks $\longrightarrow$} \\
\cmidrule(l){2-9}
&$\mathbf{b_{1}}$ &$\mathbf{b_{2}}$ &$\mathbf{b_{3}}$ &$\mathbf{b_{4}}$ &$\mathbf{b_{5}}$ &$\mathbf{b_{6}}$ &$\mathbf{b_{7}}$ &$\mathbf{b_{8}}$ \\
\midrule
\bfseries $\mathbf{L_{1}}$
&\colorbox{light-gray}{\textcolor{blue}{0}} &\colorbox{light-gray}{\textcolor{blue}{4}} &\colorbox{light-gray}{8} &\colorbox{light-gray}{13} &\colorbox{light-gray}{6} &\colorbox{light-gray}{10} &\colorbox{light-gray}{12} &-1 \\
\midrule 
\bfseries $\mathbf{L_{2}}$
&\colorbox{light-gray}{9} &\colorbox{light-gray}{\textcolor{blue}{0}} &\colorbox{light-gray}{\textcolor{blue}{4}} &\colorbox{light-gray}{\textcolor{blue}{8}} &\colorbox{light-gray}{13} &\colorbox{light-gray}{14} &\colorbox{light-gray}{2} &-1 \\
\midrule
\bfseries $\mathbf{L_{3}}$
&\colorbox{light-gray}{15} &\colorbox{light-gray}{9} &\colorbox{light-gray}{\textcolor{blue}{0}} &\colorbox{light-gray}{\textcolor{blue}{4}} &\colorbox{light-gray}{\textcolor{blue}{8}} &\colorbox{light-gray}{5} &\colorbox{light-gray}{14} &-1 \\
\midrule
\bfseries $\mathbf{L_{4}}$
&\colorbox{light-gray}{7} &\colorbox{light-gray}{15} &\colorbox{light-gray}{16} &\colorbox{light-gray}{\textcolor{blue}{0}} &\colorbox{light-gray}{\textcolor{blue}{4}} &\colorbox{light-gray}{\textcolor{blue}{8}} &\colorbox{light-gray}{1} &-1 \\
\midrule
\bfseries $\mathbf{L_{5}}$
&\colorbox{light-gray}{17} &\colorbox{light-gray}{7} &\colorbox{light-gray}{3} &\colorbox{light-gray}{16} &\colorbox{light-gray}{\textcolor{blue}{0}} &\colorbox{light-gray}{4} &\colorbox{light-gray}{\textcolor{blue}{8}} &-1 \\
\midrule
\bfseries $\mathbf{L_{6}}$
&\colorbox{light-gray}{6} &\colorbox{light-gray}{17} &\colorbox{light-gray}{18} &\colorbox{light-gray}{11} &-1 &\colorbox{light-gray}{0} &\colorbox{light-gray}{4} &\colorbox{light-gray}{8} \\
\midrule
\bfseries $\mathbf{L_{7}}$
&\colorbox{light-gray}{19} &\colorbox{light-gray}{6} &\colorbox{light-gray}{\textcolor{blue}{0}} &\colorbox{light-gray}{\textcolor{blue}{8}} &\colorbox{light-gray}{1} &\colorbox{light-gray}{2} &\colorbox{light-gray}{18} &\colorbox{light-gray}{12} \\
\midrule
\bfseries $\mathbf{L_{8}}$
&\colorbox{light-gray}{4} &\colorbox{light-gray}{19} &\colorbox{light-gray}{5} &\colorbox{light-gray}{\textcolor{blue}{\textcolor{blue}{0}}} &\colorbox{light-gray}{\textcolor{blue}{8}} &\colorbox{light-gray}{20} &\colorbox{light-gray}{\textcolor{blue}{10}} &-1 \\
\midrule
\bfseries $\mathbf{L_{9}}$
&\colorbox{light-gray}{21} &\colorbox{light-gray}{\textcolor{blue}{4}} &\colorbox{light-gray}{11} &\colorbox{light-gray}{5} &\colorbox{light-gray}{\textcolor{blue}{0}} &\colorbox{light-gray}{\textcolor{blue}{8}} &\colorbox{light-gray}{20} &-1 \\
\midrule
\bfseries $\mathbf{L_{10}}$
&\colorbox{light-gray}{1} &\colorbox{light-gray}{21} &\colorbox{light-gray}{\textcolor{blue}{4}} &\colorbox{light-gray}{3} &\colorbox{light-gray}{22} &\colorbox{light-gray}{9} &\colorbox{light-gray}{\textcolor{blue}{8}} &-1 \\
\midrule
\bfseries $\mathbf{L_{11}}$
&\colorbox{light-gray}{\textcolor{blue}{0}} &\colorbox{light-gray}{1} &\colorbox{light-gray}{23} &\colorbox{light-gray}{\textcolor{blue}{4}} &\colorbox{light-gray}{3} &\colorbox{light-gray}{22} &\colorbox{light-gray}{\textcolor{blue}{10}} &-1 \\
\midrule
\bfseries $\mathbf{L_{12}}$
&\colorbox{light-gray}{8} &\colorbox{light-gray}{\textcolor{blue}{0}} &\colorbox{light-gray}{2} &\colorbox{light-gray}{23} &\colorbox{light-gray}{\textcolor{blue}{4}} &\colorbox{light-gray}{12} &\colorbox{light-gray}{7} &\colorbox{light-gray}{11} \\
\bottomrule
\addlinespace
\end{tabular}
}
\caption{Rearranged Block Index Matrix $\boldsymbol\beta_I^\prime$ used for our work, showing the valid blocks (highlighted) to be processed.}
\label{tab:rebeta}
\end{table}
The idea is to rearrange the $\boldsymbol{\beta_I}$ matrix elements from their original order. If $\boldsymbol{\beta_I}(u,w) = \boldsymbol{\beta_I}(u^\prime,w)$, $u < u^\prime$ then \emph{stagger} the execution of $\boldsymbol{\beta_I}(u^\prime,w)$ with respect to $\boldsymbol{\beta_I}(u,w)$ by placing $\boldsymbol{\beta_I}(u^\prime,w)$ in $\boldsymbol{\beta_I^\prime}(u^\prime,w^\prime)$ such that, $w < w^\prime$. To understand how layers are pipelined, let us first look at the non-pipelined case. \\
\indent Without loss of generality, Fig. \ref{fig:pplblk}(a) shows the block-level view of the NPU timing diagram without the pipelining of layers. As seen in Section \ref{subsec:npu}, the GNPU and LNPU operate in tandem and in that order, implying that the LNPU has to wait for the GNPU updates to finish.
The layer-level picture is depicted in Fig. \ref{fig:ppllay}(a). We call this version as the $1x$ version. This idling of the GNPU and LNPU can be avoided by introducing pipelined processing of blocks given by the following Lemma.
\begin{lemma} \label{lem:ppl}
Within a superlayer, while the LNPU processes messages for the blocks $\boldsymbol{\beta^\prime}(u,w)$, the GNPU can process messages for the blocks $\boldsymbol{\beta^\prime}(u+1,w)$, $u=\{1,2,\ldots,|\mathcal{L}|-1\}$ and $w=\{1,2,\ldots,J\}$.
\end{lemma}
\begin{proof}
Follows directly from the layer independence condition in Fact \ref{fac:depcond}.
\end{proof} 
\noindent Fig. \ref{fig:pplblk}(c) illustrates the block-level view of this 2-layer pipelining scheme. It is important to note that, the splitting of the NPU process into two parts, namely, the GNPU and the LNPU (that work in tandem) is a necessary condition for Fact \ref{fac:betadep} (and hence Lemma \ref{lem:ppl}) to hold. However, at the boundary of the superlayer the Lemma \ref{lem:ppl} does not hold and pipelining has to be restarted for the next layer as seen in the layer-level view shown in Fig. \ref{fig:ppllay}(c). We call this version as the $2x$ version. This is the classical pipelining overhead. In the following, we impose certain constraints on the size of the superlayers in $\mathbf{H}$.

\begin{defin}
Without loss of generality, the pipelining efficiency $\eta_p$ is the number of layers processed per unit time per NPU array. 
\end{defin}
\noindent For the case of pipelining two layers shown in Fig. \ref{fig:ppllay}(c),
\begin{align}
\eta_p^{(2)} = \frac{|\mathcal{L}|}{|\mathcal{L}|+1}
\end{align}
Thus, we impose the following conditions on $|\mathcal{L}|$:
\begin{enumerate}
\item Since, two layers are processed in the pipeline at any given time, provided that $I$ is even,
\begin{align*}
|\mathcal{L}| \in \mathcal{F}=\{x: x \text{\, is an even factor of \,} I \}.
\end{align*}
It is important to note that, for any value of $|\mathcal{L}| \in \mathcal{F}$, $\mathcal{L}$ must be a superlayer.   
\item Given a QC-LDPC code, $|\mathcal{L}|$ is a constant. This is to facilitate a symmetric pipelining architecture which is a scalable solution.
\item Choice of $|\mathcal{L}|$ should maximize pipelining efficiency $\eta_p$,
\begin{align*}
l^* = \argmax_{|\mathcal{L}| \in \mathcal{F}} \eta_p
\end{align*}   
\end{enumerate}

\emph{Case Study}: Table \ref{tab:rebeta} shows one such rearrangement of $\boldsymbol{\beta_I}$ for the QC-LDPC code for our case study in Table \ref{tab:betai}. Unresolved dependencies are shown in blue in Table \ref{tab:rebeta}. $I=m_b=12$, $\mathcal{F}=\{2,4,6\}$ and, $l^* = \argmax_{|\mathcal{L}| \in \mathcal{F}} \eta_p = 6.$
The rearranged block index matrix $\boldsymbol{\beta_I^\prime}$ is shown in Table \ref{tab:rebeta} and the layer-level view of the pipeline timing diagram for the same is shown in Fig. \ref{fig:ppllay}(d). \\

\begin{figure*}
\centering
\includegraphics[width=1\linewidth]{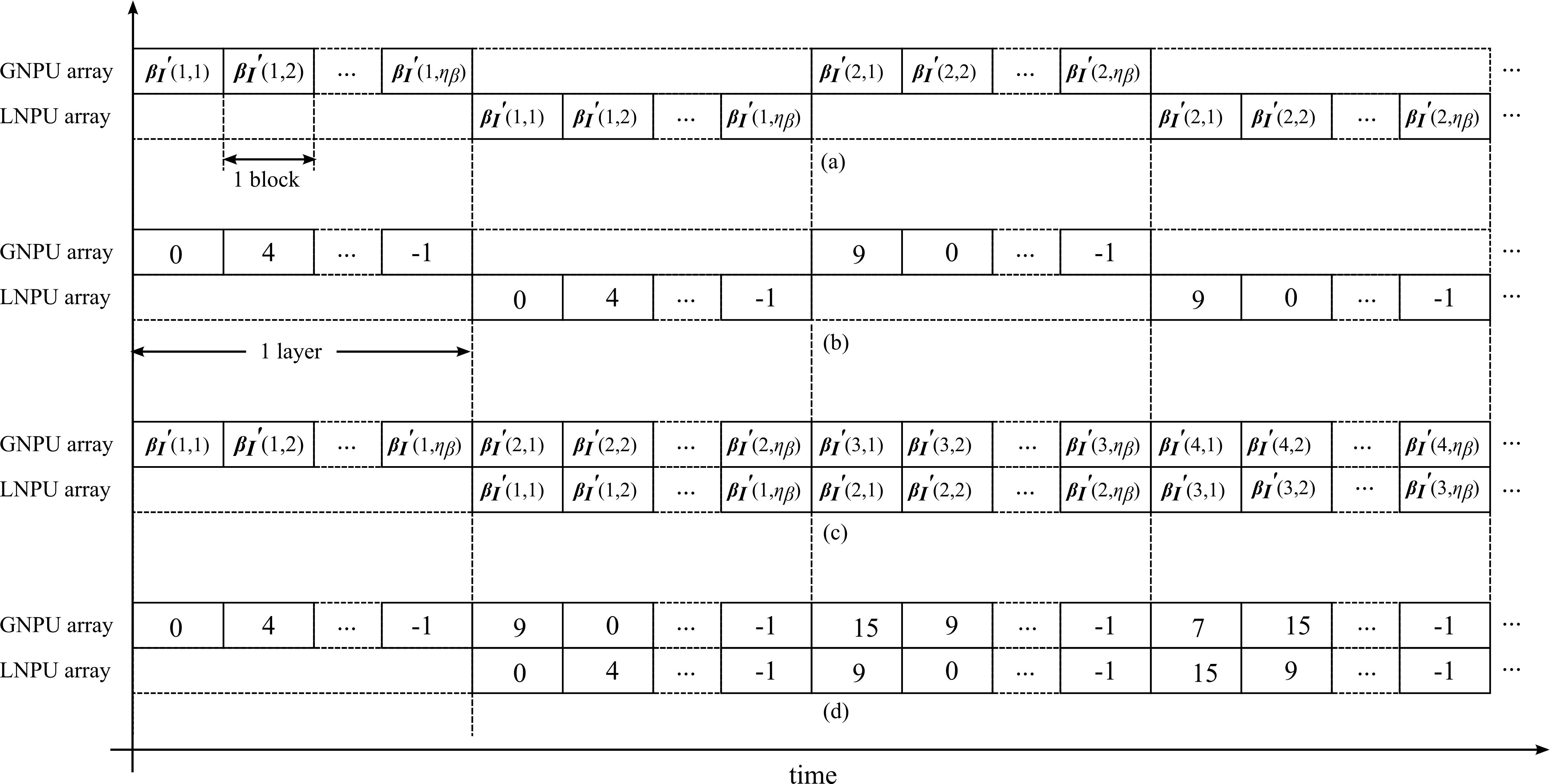}
\caption{Block-level view of the pipeline timing diagram. (a) General case for a circulant-$1$ identity submatrix construction based QC-LDPC code (see Section \ref{sec:qcldpc}) without pipelining.  (b) Special case of the IEEE 802.11n QC-LDPC code used in this work without pipelining (c) Pipelined processing of two layers for the general QC-LDPC code case in (a). (d) Pipelined processing of two layers for the IEEE802.11n QC-LDPC code case in (b).}
\label{fig:pplblk}
\end{figure*}

\begin{figure*}
\centering
\includegraphics[width=1\linewidth]{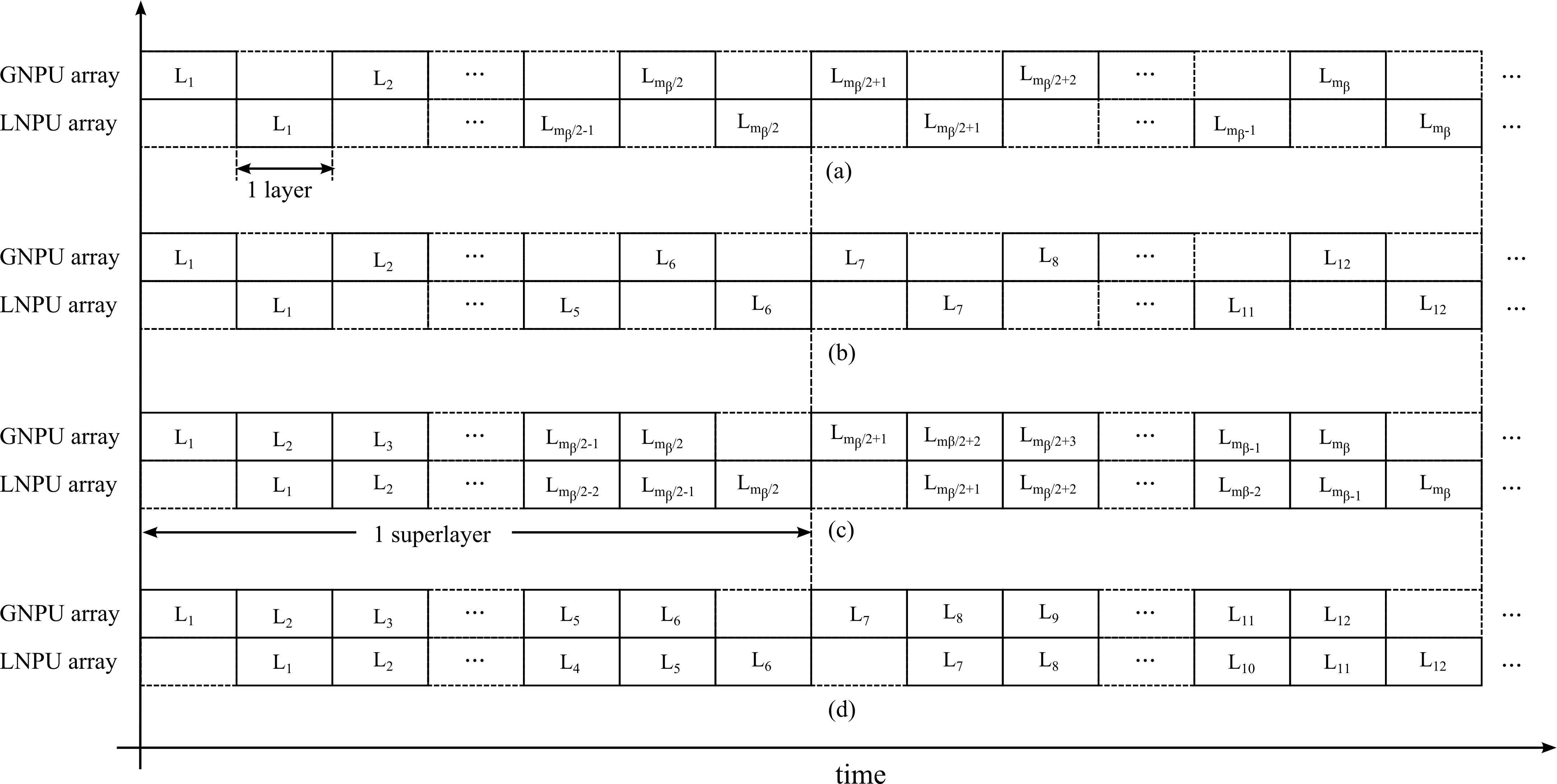}
\caption{Layer-level view of the pipeline timing diagram. (a) General case for a circulant-$1$ identity submatrix construction based QC-LDPC code (see Section \ref{sec:qcldpc}) without pipelining.  (b) Special case of the IEEE 802.11n QC-LDPC code used in this work without pipelining (c) Pipelined processing of two layers for the general QC-LDPC code case in (a). (d) Pipelined processing of two layers for the IEEE802.11n QC-LDPC code case in (b).}
\label{fig:ppllay}
\end{figure*}

\noindent \emph{High-level FPGA-based Decoder Architecture}: The high-level decoder architecture is shown in Fig. \ref{fig:decarch}. The ROM holds the LDPC code parameters specified by the $\boldsymbol\beta_I^\prime$ and the $\boldsymbol\beta_s^\prime$ along with other code parameters such as the block length and the maximum number of decoding iterations. The APP memory is initialized with the channel LLR values corresponding to all the VNs as per equation (\ref{eq:initllr}). The barrel shifter operates on blocks of VNs (APP values in equation (\ref{eq:app})) of size $z \times f$, where $f$ is the fixed-point word length used in the implementation for APP values. It circularly rotates the values to the right by using the shift values from the $\boldsymbol\beta_s^\prime$ matrix in the ROM, effectively implementing the connections between the CNs and VNs. The cyclically shifted APP memory values and the corresponding CN message values for the block in question are fed to the NPU arrays. Here, the GNPUs compute VN messages as per equation (\ref{eq:vnmsg}) and the LNPUs compute CN messages as per equation (\ref{eq:cnmsg}). These messages are then stored back at their respective locations in the RAMs for processing the next block.

\begin{figure*}
\centering
\includegraphics[width=\linewidth]{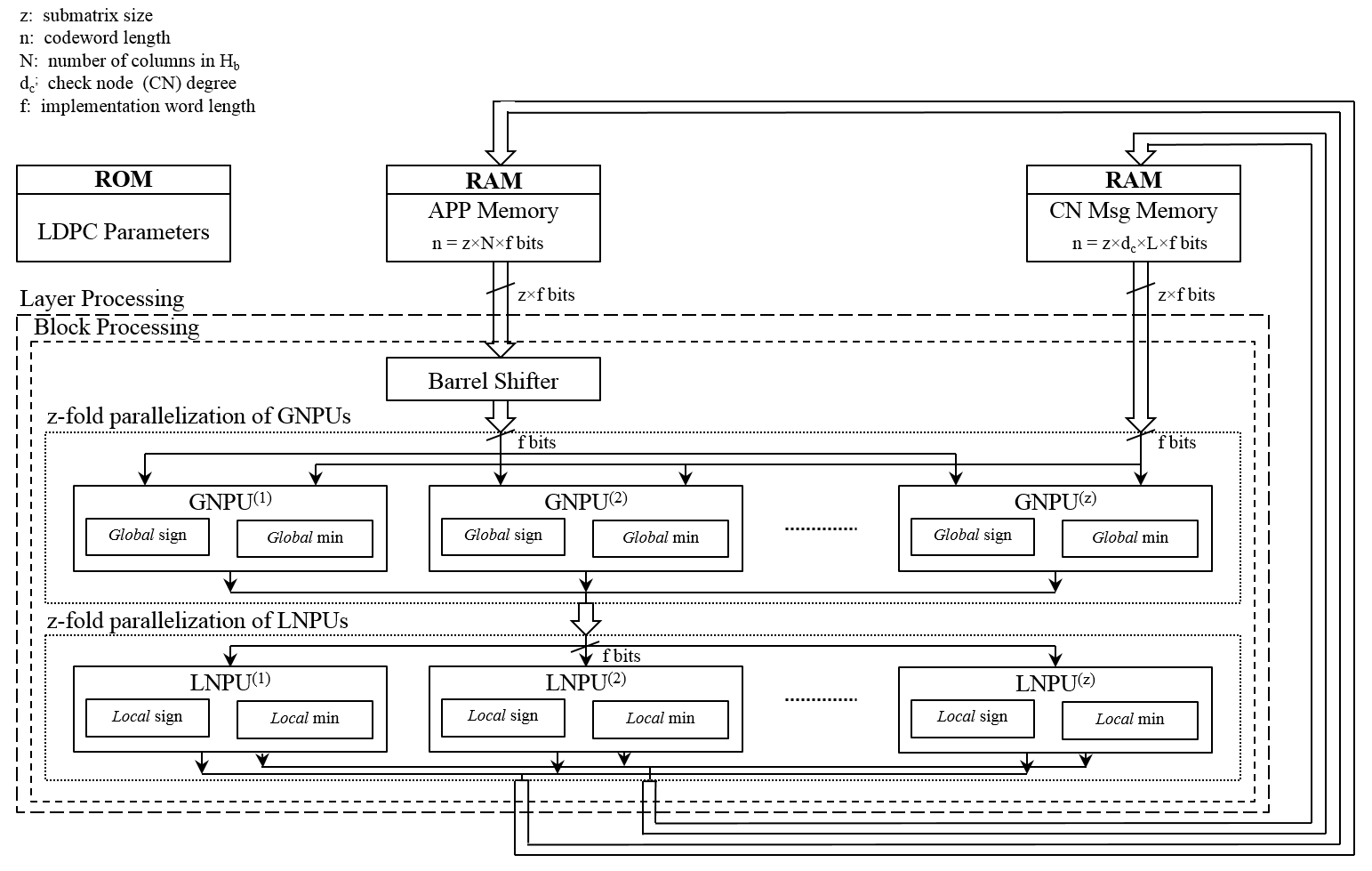}
\caption{High-level decoder architecture.\ showing the z-fold parallelization of the NPUs with an emphasis on the splitting of the sign and the minimum computation given in equation (\ref{eq:cnmsg}). Note that, other computations in equations (\ref{eq:initllr})-(\ref{eq:app}) are not shown for simplicity here. For both the pipelined and the non-pipelined versions, processing schedule for the inner Block Processing loop is as per Fig. \ref{fig:pplblk} and that for the outer Layer Processing loop is as per Fig. \ref{fig:ppllay}.}
\label{fig:decarch}
\end{figure*}

\section{Case Study}
\label{sec:casestudy}
To evaluate the proposed strategies for achieving high-throughput, we have implemented the scaled-MSA based decoder for the QC-LDPC code in the IEEE 802.11n (2012). For this code, $m_b \times n_b = 12 \times 24$, $z=27$, $54$ and $81$ resulting in code lengths of $n=24 \times z=648$, $1296$ and $1944$ bits respectively. Our implementation supports the submatrix size of $z=81$ and hence is capable of supporting all the block lengths for the rate $R=\frac{1}{2}$ code. At the time of writing this paper, we have successfully implemented the two aforementioned versions.
\subsubsection{1x}
\label{subsubsec:1x}
The block-level and the layer-level view of the pipelining is illustrated in Fig. \ref{fig:pplblk}(b) and \ref{fig:ppllay}(b) respectively.
\subsubsection{2x}
\label{subsubsec:2x}
Pipelining is done in software at the algorithmic description level. The block and layer level views of the pipelined processing are shown in Fig. \ref{fig:pplblk}(d) and \ref{fig:ppllay}(d) respectively. With an efficiency $\eta_p^{(2)}=0.86$, the \emph{2x} version is $1.7$ times faster than the \emph{1x} version.

\indent We represent the input LLRs from the channel and the CTV and VTC messages with 6 signed bits and 4 fractional bits. Fig. \ref{fig:ber} shows the bit-error rate (BER) performance for the floating-point and the fixed-point data representation with 8 decoding iterations. As expected, the fixed-point implementation suffers by about 0.5dB compared to the floating point version. The decoder algorithm was described using the \emph{LabVIEW CSDS} software. The \emph{FPGA IP} compiler was then used to generate the VHDL code from the graphical dataflow description. The VHDL code was synthesized, placed and routed using the \emph{Xilinx Vivado} compiler on the \emph{Xilinx Kintex-7} FPGA available on the \emph{NI PXIe-7975R} FPGA board. The decoder core achieves an overall throughput of 608Mb/s at an operating frequency of 200MHz and a latency of 5.7\textmu s. Table \ref{tab:results} shows that the resource usage for the \emph{2x} version (almost twice as fast due to pipelining) is close to that of the \emph{1x} version. The FPGA IP compiler chooses to use more FF for data storage in the \emph{1x} version, while it uses more BRAM in \emph{2x} version.
Compared to a contemporary FPGA-based implementation in \cite{vivadoldpc} using high-level algorithmic description compiled to an HDL, our implementation achieves a higher throughput with relatively lesser resource utilization. Authors of \cite{vivadoldpc} have implemented a decoder for a $R=\frac{1}{2}$, $n=648$, IEEE 802.11n (2012) code that achieves a throughput of $13.4$Mb/s at $122$MHz, utilizes $2\%$ of slice registers, $3\%$ of slice LUTs and $20.9\%$ of Block RAMs on the \emph{Spartan-6 LX150T} FPGA with a comparable BER performance.

\begin{figure}
\centering
\includegraphics[width=\linewidth]{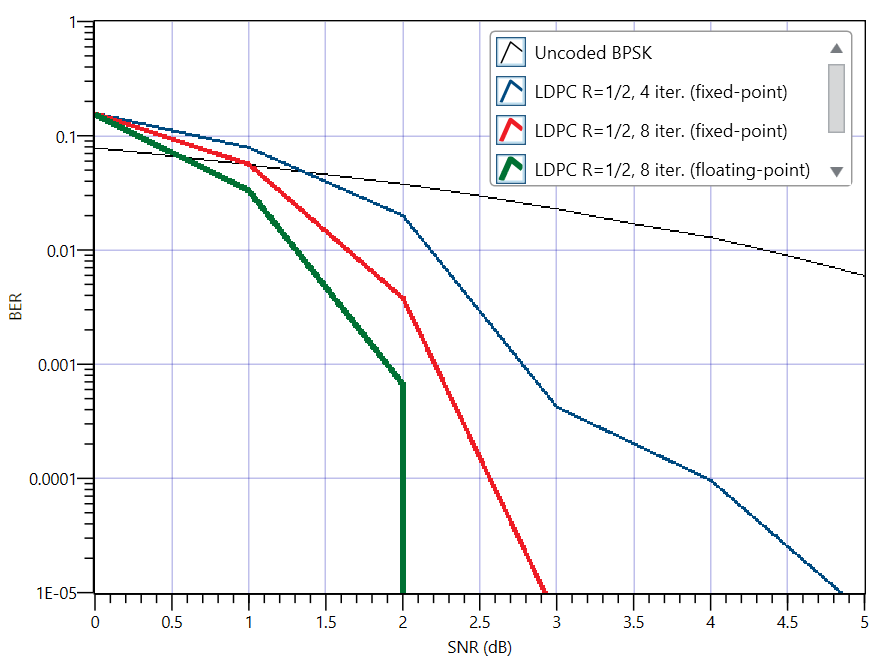}
\caption{Bit Error Rate (BER) performance comparison between uncoded BPSK (rightmost), rate=1/2 LDPC with 4 iterations using fixed-point data representation (second from right), rate=1/2 LDPC with 8 iterations using fixed-point data representation (third from right), rate=1/2 LDPC with 8 iterations using floating-point data representation (leftmost).}
\label{fig:ber}
\end{figure}

\noindent \emph{2.06 Gb/s LDPC Decoder}\cite{globecomm}: An application of this work has been demonstrated in \emph{IEEE GLOBECOM'14} where the QC-LDPC code for our case study was decoded with a throughput of $2.06$ Gb/s. This throughput was achieved by using five decoder cores in parallel on the \emph{Xilinx K7 (410t)} FPGA in the NI USRP-2953R.

\begin{table}[htp]
\centering
\setlength{\tabcolsep}{3.5pt}
\scalebox{1}{
\begin{tabular}{l@{\hspace{4pt}}*{3}{c}}
\midrule
&$\mathbf{1x}$ &$\mathbf{2x}$ \\
\midrule
\bfseries Device
& \emph{Kintex-7k410t} &\emph{Kintex-7k410t} \\
\midrule 
\bfseries Throughput(Mb/s)
&337 &608	 \\
\midrule
\bfseries FF(\%)
&9.1 &5.3 \\
\midrule
\bfseries BRAM(\%)
&4.7 &6.4 \\
\midrule
\bfseries DSP48(\%)
&5.2 &5.2 \\
\midrule
\bfseries LUT(\%)
&8.7 &8.2 \\
\bottomrule
\addlinespace \addlinespace
\end{tabular}
}
\caption{LDPC Decoder IP FPGA Resource Utilization \& Throughput on the Xilinx \emph{Kintex-7 FPGA}.}
\label{tab:results}
\end{table}

\section{Conclusion}
\label{sec:conc}
In this brief we have proposed techniques to achieve high-throughput performance for a MSA based decoder for QC-LDPC codes. The proposed compact representation of the PCM provides significant improvement in throughput. An IEEE 802.11n (2012) decoder is implemented which attains a throughput of 608Mb/s (at 260MHz) and a latency of 5.7\textmu s on the \emph{Xilinx Kintex-7} FPGA. The \emph{FPGA IP} compiler greatly reduces prototyping time and is capable of implementing complex signal processing algorithms. There is undoubtedly more scope for improvement, nevertheless, our current results are promising.

\section*{Acknowledgment}
The authors would like to thank the Department of Electrical \& Computer Engineering, Rutgers University for their continual support for this research work and the LabVIEW FPGA R\&D and the Advanced Wireless Research team in National Instruments for their valuable feedback and support.

\bibliographystyle{IEEEtran}
\bibliography{IEEEabrv,dec_arch_bib}

\begin{thebibliography}{10}
\providecommand{\url}[1]{#1}
\csname url@samestyle\endcsname
\providecommand{\newblock}{\relax}
\providecommand{\bibinfo}[2]{#2}
\providecommand{\BIBentrySTDinterwordspacing}{\spaceskip=0pt\relax}
\providecommand{\BIBentryALTinterwordstretchfactor}{4}
\providecommand{\BIBentryALTinterwordspacing}{\spaceskip=\fontdimen2\font plus
\BIBentryALTinterwordstretchfactor\fontdimen3\font minus
  \fontdimen4\font\relax}
\providecommand{\BIBforeignlanguage}[2]{{%
\expandafter\ifx\csname l@#1\endcsname\relax
\typeout{** WARNING: IEEEtran.bst: No hyphenation pattern has been}%
\typeout{** loaded for the language `#1'. Using the pattern for}%
\typeout{** the default language instead.}%
\else
\language=\csname l@#1\endcsname
\fi
#2}}
\providecommand{\BIBdecl}{\relax}
\BIBdecl

\bibitem{std80211n2012}
``{IEEE Std. for Information Technology--Telecommunications and information
  exchange between LAN and MAN--Part 11: Wireless LAN Medium Access Control
  (MAC) and Physical Layer (PHY) Specifications},'' \emph{IEEE
  P802.11-REVmb/D12, Nov 2011}, pp. 1--2910.

\bibitem{algcmp}
H.~Kee, S.~Mhaske, D.~Uliana, A.~Arnesen, N.~Petersen, T.~L. Riche, D.~Blasig,
  and T.~Ly, ``Rapid and high-level constraint-driven prototyping using
  {L}ab{VIEW} {FPGA},'' in \emph{2014 {IEEE} , GlobalSIP 2014}, 2014.

\bibitem{5g_nsn_cudak}
M.~Cudak, A.~Ghosh, T.~Kovarik, R.~Ratasuk, T.~Thomas, F.~Vook, and P.~Moorut,
  ``{Moving Towards Mmwave-Based Beyond-4G (B-4G) Technology},'' in \emph{IEEE
  77th VTC Spring '13}, June 2013, pp. 1--5.

\bibitem{ecc_shulin}
D.~Costello and S.~Lin, \emph{{Error control coding}}.\hskip 1em plus 0.5em
  minus 0.4em\relax Pearson, 2004.

\bibitem{cc_ryan}
W.~Ryan and S.~Lin, \emph{{Channel Codes: Classical and Modern}}.\hskip 1em
  plus 0.5em minus 0.4em\relax Cambridge University Press, 2009.

\bibitem{laydecarch_sun}
Y.~Sun and J.~Cavallaro, ``{VLSI Architecture for Layered Decoding of QC-LDPC
  Codes With High Circulant Weight},'' \emph{IEEE Transactions on VLSI
  Systems}, vol.~21, no.~10, pp. 1960--1964, Oct 2013.

\bibitem{htqcldpcdec_zhang}
K.~Zhang, X.~Huang, and Z.~Wang, ``{High-throughput layered decoder
  implementation for QC-LDPC codes},'' \emph{IEEE Journal on Selected Areas in
  Communications}, vol.~27, no.~6, pp. 985--994, Aug 2009.

\bibitem{fpar_onizawa}
N.~Onizawa, T.~Hanyu, and V.~Gaudet, ``Design of high-throughput fully parallel
  ldpc decoders based on wire partitioning,'' \emph{IEEE Transactions on VLSI
  Systems}, vol.~18, no.~3, pp. 482--489, Mar 2010.

\bibitem{lcmpdec_mohsenin}
T.~Mohsenin, D.~Truong, and B.~Baas, ``A low-complexity message-passing
  algorithm for reduced routing congestion in ldpc decoders,'' \emph{IEEE
  Transactions on Circuits and Systems I: Regular Papers}, vol.~57, no.~5, pp.
  1048--1061, May 2010.

\bibitem{htfpga_balatsoukas}
A.~Balatsoukas-Stimming and A.~Dollas, ``{FPGA-based design and implementation
  of a multi-Gbps LDPC decoder},'' in \emph{International Conference on Field
  Programmable Logic and Applications (FPL)}, Aug 2012, pp. 262--269.

\bibitem{2bitmsa_chandrasetty}
V.~Chandrasetty and S.~Aziz, ``{FPGA Implementation of High Performance LDPC
  Decoder Using Modified 2-Bit Min-Sum Algorithm},'' in \emph{International
  Conference on Computer Research and Development}, May 2010, pp. 881--885.

\bibitem{mgbpsfpga_wilson}
R.~Zarubica, S.~Wilson, and E.~Hall, ``Multi-gbps fpga-based low density parity
  check (ldpc) decoder design,'' in \emph{IEEE GLOBECOM '07}, Nov 2007, pp.
  548--552.

\bibitem{gbpsdecovw_schlafer}
P.~Schl{\"a}fer, C.~Weis, N.~Wehn, and M.~Alles, ``Design space of flexible
  multigigabit ldpc decoders,'' \emph{VLSI Design}, vol. 2012, p.~4, 2012.

\bibitem{vivadoldpc}
E.~Scheiber, G.~H. Bruck, and P.~Jung, ``{Implementation of an LDPC decoder for
  IEEE 802.11n using Vivado TM High-Level Synthesis},'' in \emph{{International
  Conference on Electronics, Signal Processing and Communication Systems}},
  2013.

\bibitem{ht_impl}
S.~Mhaske, D.~Uliana, H.~Kee, T.~Ly, A.~Aziz, and P.~Spasojevic, ``{A 2.48Gb/s
  QC-LDPC Decoder Implementation on the NI USRP-2953R},'' in \emph{Vehicular
  Technology Conference (VTC Fall), 2015 IEEE 82nd}, Sep 2015, pp. 1--5,
  submitted for publication.

\bibitem{ldpc_gallager}
R.~G. Gallager, ``Low-density parity-check codes,'' \emph{Information Theory,
  IRE Transactions on}, vol.~8, no.~1, pp. 21--28, 1962.

\bibitem{ldpc_tanner}
R.~Tanner, ``{A recursive approach to low complexity codes},''
  \emph{Information Theory, IEEE Transactions on}, vol.~27, no.~5, pp.
  533--547, Sep 1981.

\bibitem{spa_factorgraphs}
F.~Kschischang, B.~Frey, and H.-A. Loeliger, ``{Factor graphs and the
  sum-product algorithm},'' \emph{Information Theory, IEEE Transactions on},
  vol.~47, no.~2, pp. 498--519, Feb 2001.

\bibitem{serialmp_litsyn}
E.~Sharon, S.~Litsyn, and J.~Goldberger, ``{Efficient Serial Message-Passing
  Schedules for LDPC Decoding},'' \emph{IEEE Transactions on Information
  Theory}, vol.~53, no.~11, pp. 4076--4091, Nov 2007.

\bibitem{laydec}
M.~Mansour and N.~Shanbhag, ``High-throughput {LDPC} decoders,'' \emph{{IEEE
  Transactions on VLSI Systems}}, vol.~11, no.~6, pp. 976--996, Dec 2003.

\bibitem{smsa_chen}
J.~Chen and M.~Fossorier, ``Near optimum universal belief propagation based
  decoding of ldpc codes and extension to turbo decoding,'' in \emph{IEEE ISIT
  '01}, 2001, p. 189.

\bibitem{valreuse_gunnam}
K.~Gunnam, G.~Choi, M.~Yeary, and M.~Atiquzzaman, ``{VLSI Architectures for
  Layered Decoding for Irregular LDPC Codes of WiMax},'' in \emph{{IEEE ICC
  '07}}, June 2007, pp. 4542--4547.

\bibitem{globecomm}
H.~Kee, D.~Uliana, A.~Arnesen, N.~Petersen, T.~Ly, A.~Aziz, S.~Mhaske, and
  P.~Spasojevic, ``{Demonstration of a 2.06Gb/s LDPC Decoder},'' in \emph{{IEEE
  GLOBECOM '14}}, Dec 2014, \url{https://www.youtube.com/watch?v=o58keq-eP1A }.

\end{thebibliography}

\end{document}